\documentclass[10pt,twocolumns]{IEEEtran}
\usepackage{amsmath, graphics,amssymb,epsfig,subfigure,color,amsthm,cite}

\newtheorem{theorem}{Theorem}

\newtheorem{corollary}{Corollary}
\newtheorem{definition}{Definition}
 
\begin{document}
 
\title{{ Index Coding with Coded Side-Information} 
}

\author{
  \IEEEauthorblockN{Namyoon Lee, Alexandros G. Dimakis, and Robert W. Heath Jr. }
\thanks{N. Lee, A. G. Dimakis, and R. W. Heath Jr. are with the Wireless Networking and Communications Group, Department of Electrical and Computer Engineering, The University of Texas at
Austin, Austin, TX 78712, USA. (e-mail:\{namyoon.lee, rheath\}@utexas.edu and dimakis@austin.utexas.edu). This work was supported in part by Intel Labs.}  
}
 
 
\maketitle

\begin{abstract}
This letter investigates a new class of index coding problems. One sender broadcasts packets to multiple users, each desiring a subset, by exploiting prior knowledge of linear combinations of packets. We refer to this class of problems as index coding with \textit{coded} side-information. Our aim is to characterize the minimum index code length that the sender needs to transmit to simultaneously satisfy all user requests. We show that the optimal binary vector index code length is equal to the minimum rank (minrank) of a matrix whose elements consist of the sets of desired packet indices and side-information encoding matrices. This is the natural extension of matrix minrank in the presence of coded side information. Using the derived expression, we propose a greedy randomized algorithm to minimize the rank of the derived matrix.
 
\end{abstract}

\begin{keywords} 
Index coding and coded side-information\end{keywords} 
 
\section{Introduction}
Index coding in \cite{Birk,Yossef,Alon} is a transmission technique for noiseless broadcasting channel consisting of a transmitter and a set of users. The transmitter wishes to deliver multiple packets to their respective users over a shared noiseless link. Each user has its own prior knowledge of a subset of the packets. The transmitter sends a signal per time slot and all the users receive it without noise. The goal is to design transmit codes to minimize the number of required transmissions so that all users decode the desired packets with their own side-information and the received signals from the transmitter. This class of problems has recently 
received attention because of its connections to network coding \cite{Effros} and topological interference management \cite{Jafar}. Designing an efficient index code is tightly related with the constructing codes for caching \cite{MN} and distributed storage systems \cite{Dimakis}.

There has been extensive work on characterizing the optimal index code length (the minimum number of transmissions) \cite{Birk,Yossef, Alon}. Approaches based on graph theory are popular because of the strong connection between the optimal index code length and graph-theoretical quantities \cite{Birk,
Alon}. For instance, when each user wants distinct packets, an index code design problem is equivalently represented in terms of a directed side-information graph $G$. It was shown in \cite{Yossef} that, for the given direct side-information graph $G$, the optimal index code length is lower and upper bounded by the maximum independent set number of the corresponding graph, $\alpha(G)$, and the chromatic number of its complement, $\chi(\bar{G})$. These approaches  \cite{Birk,Yossef, Alon} are useful in characterizing the bounds of the optimal index code length and in obtaining the optimal index code for a certain class of side-information graphs (e.g., vertex-coloring methods \cite{Dimakis2}).

Algebraic approaches are also effective methods to characterize the optimal index code length. One key result is that the optimal binary index code length equals to the minimum rank (minrank) of a matrix that fits the side-information graph $G$, i.e., ${\rm minrk}(G)$ \cite{Yossef}. This algebraic expression yields a new way of constructing index codes by solving a matrix completion problem over a finite field.

In this letter, we consider a generalization of index coding when the side information packets
can be themselves coded. Specifically, unlike the conventional assumption that each user independently knows a subset of other users' packets as side-information \cite{Birk,Yossef,Alon,Dimakis,Dimakis2}, we admit a coded structure in generating side-information so that each user is able to have any linear combinations of all packets as side-information. As is well-known, index coding is already tremendously challenging even the side information packets are not coded. However, the idea of allowing coded side information is useful (beyond mathematical interest as a natural generalization) since in many situations side information is created by overhearing previous transmissions which will be very frequently coded. This is especially relevant in caching scenarios where helpers try to assist in content dissemination \cite{MN}.

To explain the index coding problem with coded side-information, we provide a motivating example. As depicted in Fig. \ref{fig:1}, a transmitter desires to deliver $a$, $b$, and $c$ to user 1, 2, and 3, respectively. Let us first consider the uncoded side-information case  where each user separately knows the others' desired information bits. In such case, one XORed transmission $a+b+c$ suffices to make all three users decode the desired information bits, if user 1, 2, and 3 stored two bits $\{b,c\}$, $\{a,c\}$, and $\{a,b\}$, where $+$ represents an XOR operation or addition over the binary field $\mathbb{F}_2$. Next, let us consider a different scenario in which each user may only store one bit due to the lack of memory in the device. In this case, if user 1, 2, and 3 have coded side-information of $b+c$, $a+c$, and $a+b$, the same XORed transmission $a+b+c$ are enough to satisfy all three users. This example reveals the benefit of coding over side-information in reducing the size of cache while maintaining the transmission rate. 
For the case of uncoded side-information, each user requires cache memory with two bits to decode the desired information bit from the XORed transmission $a+b+c$ by the transmitter. Whereas, if each user strores the XORed bit instead of saperately storing them, cache memory with one bit is enough to extract the desired information bit.
\begin{figure}
\centering
\includegraphics[width=3.5in]{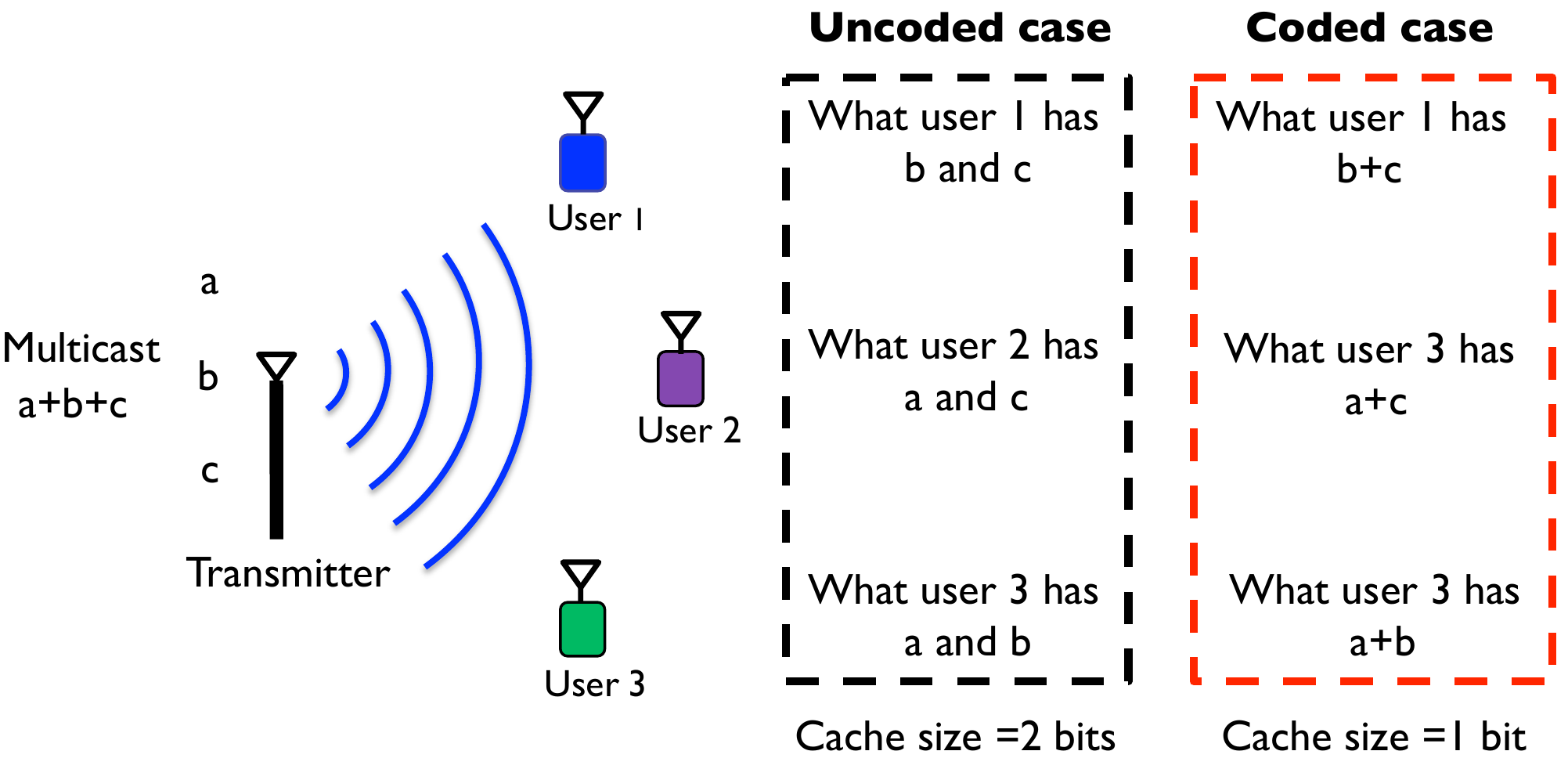}
\caption{A motivating example of the index coding problem with coded side-information. For the case of uncoded side-information, each user requires cache memory with two bits to decode the desired information bit from the XORed transmission $a+b+c$ by the transmitter. Whereas, if each user strores the XORed bit instead of saperately storing them, cache memory with one bit is enough to extract the desired information bit.} \label{fig:1}  
\end{figure}

Our main contribution is to characterize the optimal binary vector linear index code length in terms of the minrank of a matrix when the users have coded side-information. Our key finding is that the minrank expression is a function of 1) a set of the packet indices requested by the users and 2) a set of side-information encoding matrices. With the derived equivalence between the optimal index code length and the minrank expression, we propose a random greedy algorithm that minimizes the rank of the derived matrix. The index coding problem with coded side-information, in fact, was initially considered in \cite{Shum} where a linear index code with coded side-information can be found equivalently by solving a system of multi-variable polynomial equations, which is difficult to solve in general. We show how to design index codes by minimizing the rank of a derived matrix. This rank minimization formulation allows us to connect the index code design problem to a matrix completion problem over finite fields \cite{Draper}. 
 

\section{Problem Formulation}
Consider one transmitter and $K$ users in a network. The transmitter has $N$ packets, each with $F$ bits, ${\bf x}_n \in \mathbb{F}_2^F$ for $n\in\mathcal{N}=\{1,2,\ldots,N\}$. We denote a sequence of all packets by ${\bf x}=[{\bf x}_1,{\bf x}_2,\ldots, {\bf x}_N]^T\in \mathbb{F}_2^{FN}$. User $k \in\{1,2,\ldots,K\}=\mathcal{K}$ requests a set of packets $\{{\bf x}_{i}\}$ for $ i\in\mathcal{T}_k\subset \mathcal{N}$ ($\mathcal{T}_k$ is a subset of $\mathcal{N}$). For example, if $\mathcal{T}_k=\{1,2\}$, then user $k$ desires to decode packets ${\bf x}_1$ and ${\bf x}_2$. Further, user $k\in\mathcal{K}$ has coded side-information ${\bf u}_k \in \mathbb{F}_2^{M_kF}$ with the size of $M_kF$ bits for $M_kF\in\mathbb{N}$. Assuming a linear encoding method, side-information given to user $k$, ${\bf u}_k$, is created by a side-information generating matrix ${\bf S}_k\in \mathbb{F}_2^{M_kF\times NF}$ as
\begin{align}
{\bf u}_k={\bf S}_k{\bf x}.
\end{align}
With knowledge of the set of encoding matrices $\{{\bf S}_1,\ldots,{\bf S}_K\}$, the transmitter sends different linear combinations of packets ${\bf x}$ over $LF$ time slots (channel uses) so that all users successfully decode the requested packets by exploiting their coded side-information ${\bf u}_k$. 

Under the restriction of linear coding, the transmitter uses an index coding matrix as an encoding function, i.e., ${\bf C}_{\rm IC}=[{\bf C}_1,{\bf C}_2,\ldots,{\bf C}_N]\in \mathbb{F}^{LF\times NF}$. Note that the $k$-th sub-matrix ${\bf C}_{k} \in\mathbb{F}_2^{LF\times F}$ is the precoding matrix carrying the $k$-th packet ${\bf x}_k$. When the transmitter sends $L$ packets with the index coding matrix ${\bf C}_{\rm IC}$ over a noiseless link, user $k$ obtains the information vector ${\bf y}\in\mathbb{F}_2^{LF} $ over $LF$ channel uses as
\begin{align}
{\bf y}={\bf C}_{\rm IC}{\bf x}.
\end{align}
Applying a linear decoder ${\bf D}_k^T \in \mathbb{F}_2^{|\mathcal{T}_k|F\times (L+M_k)F}$, user $k \in\{1,2,\ldots,K\}$  decodes packet ${\bf x}_{i}$ for $ i\in\mathcal{T}_k$ using both the received signal ${\bf y}$ and the coded side-information vector ${\bf u}_k$. The decodablity condition at user $k$ is 
\begin{align}
{\bf D}_k^T\left[%
\begin{array}{c}
{\bf y} \\\hline
 {\bf u}_k\\
\end{array}%
\right]= {\bf R}_k{\bf x},\label{eq:dec_cond_k}
\end{align}
where ${\bf R}_k \in \mathbb{F}_2^{|\mathcal{T}_k|F\times NF}$ denotes the index matrix of the requested packets by user $k$ and ${\bf R}_k\neq {\bf S}_k$. Hence, the index matrix ${\bf R}_k$ is a block matrix whose $(i,t_i)$ sub-block is an identity matrix ${\bf I}_{F}$ if $t_i\in \mathcal{T}_k$; otherwise the remaining blocks are zero matrices. 
With the decodability condition in (\ref{eq:dec_cond_k}), we define a valid linear index code and its optimal code length.

\begin{definition} (Valid linear index code) The index coding matrix ${\bf C}_{\rm IC} \in \mathbb{F}_2^{LF\times NF}$ is valid over $\mathbb{F}_2$ with the length $LF$ if every user is able to decode its desired sets of packets from the transmitted packets and  side-information available at user $k$. In other words, all users simultaneously satisfy the decodability conditions in (\ref{eq:dec_cond_k}).
\end{definition}

\begin{definition} (Optimal linear index code length) It is said that the index coding matrix ${\bf C}_{\rm IC} \in \mathbb{F}_2^{LF\times NF}$ has the optimal length $\beta_2^{\star}$ if ${\bf C}_{\rm IC} \in \mathbb{F}_2^{LF\times NF}$ is valid and with the minimum number of rows $\beta_2^{\star}=\min LF$.
\end{definition}

\section{Main Results}
In this section, we characterize the minrank expression of the index code length for the class of index coding problems with coded side-information. The following theorem is the main result of this paper.

\begin{theorem} \label{Theorem1}For the given set of side-information generating matrices $\left\{{\bf S}_1,\ldots,{\bf S}_K\right\}$ and the desired packet index matrices $\left\{{\bf R}_1,\ldots,{\bf R}_K\right\}$, the optimal linear vector index code length $\beta^{\star}_2$ over $\mathbb{F}_2$ is obtained by solving the following optimization problem: \begin{align}
\beta^{\star}_2(\{{\bf S}_k,{\bf R}_k\}_{k=1}^K)=\!\!\!\min_{{\bf A}_1^T,\ldots,{\bf A}_K^T} \!\!\!{\rm rk}\left(\!\!\small\left[%
\begin{array}{c}
 {\bf R}_1+{\bf A}_1^T{\bf S}_1  \\ \vdots \\ {\bf R}_K+{\bf A}_K^T{\bf S}_K
\end{array}%
\right]\right),\label{eq:Th1}
\end{align}
where ${\bf A}_k^T \in \mathbb{F}_2^{|\mathcal{T}_k|F\times M_kF}$.\end{theorem}

\begin{proof} We prove Theorem \ref{Theorem1} using an algebraic approach. Recall the decodability condition of user $k\in\mathcal{K}$ in (\ref{eq:dec_cond_k}). We decompose the decoding matrix ${\bf D}_k^T$ into two sub matrices ${\bf B}_k^T \in \mathbb{F}_2^{|\mathcal{T}_k|F \times LF}$ and
 ${\bf A}_k^T \in \mathbb{F}_2^{|\mathcal{T}_k|F\times M_kF}$ as 
\begin{align}
{\bf D}_{k}^T= \left[%
\begin{array}{cc}
 {\bf B}_k^T & {\bf A}_k^T
\end{array}%
\right],
\end{align}
where sub-matrices ${\bf B}_k^T$ and ${\bf A}^T_k$ are multiplied to the received signal vector ${\bf y}$ and side-information vector ${\bf u}_k$, respectively. With these sub-matrices, the decodability condition in (\ref{eq:dec_cond_k}) at user $k$ is equivalently decomposed as
\begin{align}
{\bf B}_k^T{\bf y}+{\bf A}_k^T{\bf u}_k&={\bf R}_k{\bf x}. 
\label{eq:dec_cond_k_compact}
\end{align}
Using the fact that ${\bf y}={\bf C}_{\rm IC}{\bf x}$ and ${\bf u}_k={\bf S}_k{\bf x}$, the decodability condition in (\ref{eq:dec_cond_k_compact}) is rewritten as
\begin{align}
\left({\bf B}_k^T{\bf C}_{\rm IC}+{\bf A}_k^T{\bf S}_k\right){\bf x}&={\bf R}_k{\bf x}. \label{eq:dec_k_mat}
\end{align}
Since ${\bf x}$ is non-degenerate, the decodability condition in (\ref{eq:dec_k_mat}) simplifies further as
\begin{align}
{\bf B}_k^T{\bf C}_{\rm IC}&+{\bf A}_k^T{\bf S}_k ={\bf R}_k, \nonumber \\
{\bf B}_k^T{\bf C}_{\rm IC}&={\bf R}_k+{\bf A}_k^T{\bf S}_k,\label{eq:dec_k_mat_2}
\end{align}
where the last equality is due to the addition over $\mathbb{F}_2$.
Since every users needs to satisfy the decodability condition in (\ref{eq:dec_k_mat_2}), the decodability condition for all the users is given by
\begin{align}
\small\underbrace{\left[%
\begin{array}{c}
{\bf B}_1^T \\ \vdots \\{\bf B}_{K}^T
\end{array}%
\right]}_{(\sum_{k=1}^K|\mathcal{T}_k|)F\times LF} \!\!\!\!\underbrace{{\bf C}_{\rm IC}}_{LF\times NF}=\underbrace{\left[%
\begin{array}{c}
 {\bf R}_1+{\bf A}_1^T{\bf S}_1  \\ \vdots \\ {\bf R}_K+{\bf A}_K^T{\bf S}_K
\end{array}%
\right]}_{(\sum_{k=1}^K|\mathcal{T}_k|)F\times NF}.\label{eq:dec_cond_global}
\end{align}  
Notice that the rank of each matrix in the left-hand side in (\ref{eq:dec_cond_global}) respectively equals to $LF$. This is because, by definition, ${\bf C}_{\rm IC}$ should have $LF$ linearly independent rows as a transmitter sends out a linearly independent linear combination of packets per time slot. Furthermore, the concatenating matrix of all decoding matrices ${\bf B}^T_k$ also has $LF$ linearly independent columns to employ received vector ${\bf y}\in\mathbb{F}_2^{LF}$ in decoding. We denote the concatenating matrix of all decoding matrices by ${\bf \bar B}=\left[{\bf B}_1,\ldots, {\bf B}_K\right]^T$.
From the rank inequality, the rank of the product of the two matrices is upper bounded by
\begin{align}
 \small{\rm rk}\!\left({\bf \bar B}{\bf C}_{\rm IC}\!\right) \leq \min\left\{ {\rm rk}\!\left({\bf \bar B}\right), {\rm rk}({\bf C}_{\rm IC})\!\right\}
=LF.
\end{align}
Furthermore, applying Sylvester's rank inequality, we obtain the lower bound on the rank as
\begin{align}
 \small{\rm rk}\left({\bf \bar B}{\bf C}_{\rm IC}\!\right) \geq  {\rm rk}\left({\bf \bar B}\right)+{\rm rk}({\bf C}_{\rm IC})-LF
=LF.
\end{align} 
As a result, we conclude that the rank of the matrix in the right-hand side in (\ref{eq:dec_cond_global}) equals to $LF$, namely, \begin{align}
LF= {\rm rk}\left(\small\left[%
\begin{array}{c}
 {\bf R}_1+{\bf A}_1^T{\bf S}_1 \\ \vdots \\ {\bf R}_K+{\bf A}_K^T{\bf S}_K
\end{array}%
\right]\right). \label{eq:rankmat}
\end{align}
Since we are interested in finding the minimum $LF$, the optimal index code length $\beta_2^{\star}(\{{\bf S}_k,{\bf R}_k\}_{k=1}^K)=\min LF$ is obtained by minimizing the rank of the matrix in (\ref{eq:rankmat}) with respective to over all possible indeterminate elements in $\{{\bf A}_k\}_{k=1}^K$. Consequently, the minimum index code length is obtained by solving the optimization problem stated in (\ref{eq:Th1}).
\end{proof}

Theorem 1 shows that the optimal linear index code length is determined by two factors: 1) the set of the packet index matrices $\{{\bf R}_k\}$ and 2) the set of side-information encoding matrices $\{{\bf S}_k\}$. Furthermore, the derived minrank expression in (\ref{eq:Th1}) is useful to design the optimal index coding matrix ${\bf C}_{\rm IC}$ with the rank of $\beta_2^{\star}$. This is because, under the premise that $\{{\bf B}_k\}$ is predefined as ${\bf B}_k^T={\bf I}_{|\mathcal{T}_k|F}$ for $k\in\mathcal{K}$, it is possible to attain the optimal index coding matrix with rank $\beta_2^{\star}$, ${\bf C}_{\rm IC}^{\star}$, by arbitrary selecting  a set of the $\beta_2^{\star}$ linearly independent rows in (\ref{eq:Th1}) with $\{{\bf A}_k^{\star}\}$. Therefore, the index coding matrix can be obtained by carefully completing the indeterminate elements in $\{{\bf A}_k^{\star}\}$ so that they provide the minimum rank of the resultant matrix. This motivates us to design an algorithm that finds the index coding matrix via a matrix completion approach, which will be explained in Section IV.

To shed further light on the significance of Theorem \ref{Theorem1}, it is instructive to consider certain special cases and examples.

A special case is when $F=1$, $N=K$, and $M_k=1$. User $k\in\{1,2,\ldots,K\}$ requests packet $x_k\in\mathbb{F}_2$ with one bit file size, i.e., $\mathcal{T}_k=\{k\}$ and $|\mathcal{T}_k|=1$. Therefore, the packet requested by user $k$ is simply written as a unit vector whose $k$-th element is one, ${\bf R}_k= {\bf e}_k^T$.
Further, we assume that the memory size of user $k$ is one bit, i.e., $M_k=1$ for $\forall k$. Then, the coded side-information generating matrix becomes a vector ${\bf s}_k^T \in\mathbb{F}_2^K$. In this reduced setup, the optimal index code length is stated in the following corollary.

\begin{corollary} \label{Corollary1}
When $N=K$ and $F=M_k=1$, the optimal scalar linear index code length is obtained by solving the following optimization problem
\begin{align}
\beta^{\star}_2(\{{\bf s}_k,{\bf e}_k\}_{k=1}^K)&=\min_{{ a}_1,\ldots,{a}_K} {\rm rk}\left(\small\left[%
\begin{array}{c}
 {\bf e}_1^T+ { a}_1 {\bf s}_1^T  \\ \vdots \\ {\bf e}_K^T+{a}_K{\bf s}_K^T
\end{array}%
\right]\right) \nonumber \\
&=\min_{{\bf \bar A}} {\rm rk}\left( {\bf I}_K + {\bf \bar S}{\bf \bar A}\right), \label{eq:cor1}
\end{align}
where ${\bf \bar S}=\left[{\bf s}_1,\ldots,{\bf s}_K\right]$ and ${\bf \bar A}={\rm diag}[a_1, a_2,\ldots,a_K]^T$.
\end{corollary}
\proof Without loss of generality, we assume that user $k$ desires to decode file $x_k$, i.e., ${\bf R}_k={\bf e}_k^T$ with side-information $u_k={\bf s}_k^T{\bf x}$. Then, from Theorem 1, the optimal index code length is obtained by solving the problem stated in (\ref{eq:cor1}).
\endproof

\textbf {Example 1 (Optimal Side-Information Encoding Structure): } For the given $a_k=1$, if we choose ${\bf \bar S}={\bf J}_K+ {\bf I}_K$ where ${\bf J}_K\in\mathbb{F}_2^{K\times K}$ is a all-ones matrix, then the rank of the matrix ${\bf I}_K + {\bf \bar S}{\bf \bar A}$ becomes one as ${\bf I}_K + {\bf \bar S}{\bf \bar A}={\bf I}_K + {\bf J}_K+ {\bf I}_K=  {\bf J}_K$ and $\rm{rank}({\bf J}_k)=1$. As a result, we conclude that the optimal index code with the length one is achievable \textit{if and only if} the product of the side-information encoding matrix and the free variables in ${\bf A}$ (decoding matrix) has a particular structure of ${\bf \bar S}{\bf \bar A}={\bf J}_K+ {\bf I}_K$. This confirms the intuition that if each user knows XORed information of all packets excepting for its desired one as side-information, it is possible to satisfy all users by sending XORed information of all packets within one channel use.

\textbf {Example 2 (Connection to Index Coding with Uncoded Side-Information): } Let us consider the following index coding problem where $K=N=5$ and $F=1$. User $k\in\{1,\ldots,5\}$ desires to decode $x_k\in \mathbb{F}_2$ with the set of uncoded side-information as follows:
\begin{itemize}
\item User 1 has $x_2$ and $x_5$, i.e., ${\bf S}_1=[{\bf e}_2~ {\bf e}_5]^T \in \mathbb{F}_2^{2\times 5}$,
\item User 2 has $x_1$ and $x_5$, i.e., ${\bf S}_2=[{\bf e}_1~ {\bf e}_5]^T \in \mathbb{F}_2^{2\times 5}$,
\item User 3 has $x_2$ and $x_4$, i.e., ${\bf S}_3=[{\bf e}_2~{\bf e}_4]^T \in \mathbb{F}_2^{2\times 5}$,
\item User 4 has $x_2$ and $x_3$, i.e., ${\bf S}_4=[{\bf e}_2~{\bf e}_3]^T \in \mathbb{F}_2^{2\times 5}$,
\item User 5 has $x_1$, $x_3$, and $x_4$, i.e., ${\bf S}_5=[{\bf e}_1~ {\bf e}_3~ {\bf e}_4]^T \in \mathbb{F}_2^{3 \times 5}$.
\end{itemize}
Note that since the side-information is uncoded, each row in side-information matrix ${\bf S}_k$ contains a non-zero element. Denoting ${\bf A}_k^{T}=[a_{k}^1, a_{k}^2]\in \mathbb{F}_2^{2\times 1}$ for $k\in\{1,\ldots,4\}$ and ${\bf A}_5^T=[a_5^1, a_5^2, a_5^3]\in\mathbb{F}_2^{3\times 1}$, from Theorem 1, we are able to find the optimal index coding matrix ${\bf C}_{\rm IC}$ by solving the following optimization problem:
\begin{align}
\beta_2^{\rm uncoded}\!=\!\!\!\!\min_{\{{a}^1_k,a^2_k\},a^3_5} \!\!{\rm rk}\left(\small{\left[%
\begin{array}{ccccc}
 1 & a^1_{1} &0 &0 & a^2_{1} \\
 a^1_{2} &1&0&0 &a^2_{2}\\
 0& a^1_{3} &1& a^2_{3} &0\\
 0& a^1_{4}&a^2_{4} &1 &0\\
 a^1_5 &0 & a^2_5 & a^3_5 &1 \\
\end{array}%
\right]}\right) \label{eq:ex2_uncoded}. 
\end{align}
Now, we consider the same index coding problem but   side-information is coded as follows:
 \begin{itemize}
\item User 1 has $x_2+x_5$, i.e., ${\bf s}_1=[{\bf e}_2+{\bf e}_5]^T \in \mathbb{F}_2^{1\times 5}$,
\item User 2 has $x_1+x_5$, i.e., ${\bf s}_2=[{\bf e}_1+{\bf e}_5]^T \in \mathbb{F}_2^{1\times 5}$,
\item User 3 has $x_2+x_4$, i.e., ${\bf s}_3=[{\bf e}_2+{\bf e}_4]^T \in \mathbb{F}_2^{1\times 5}$,
\item User 4 has $x_2+x_3$, i.e., ${\bf s}_4=[{\bf e}_2+{\bf e}_3]^T \in \mathbb{F}_2^{1\times 5}$,
\item User 5 has $x_1+x_3+x_4$, i.e., ${\bf s}_5=[{\bf e}_1 + {\bf e}_3+{\bf e}_4]^T \in \mathbb{F}_2^{1\times 5}$.
\end{itemize}
Since each user has coded side-information, unlike the uncoded case, the decoding matrix for user $k$ becomes ${\bf A}_k^T=[a_k]\in\mathbb{F}_2^{1}$. As a result, using Theorem 1, the optimal index coding matrix for the coded case is obtained by solving the following optimization problem:
\begin{align}
\beta_2^{\rm coded}=\min_{{a}_1,\ldots,a_5} {\rm rk}\left(\small\left[%
\begin{array}{ccccc}
 1 & a_{1} &0 &0 & a_{1} \\
 a_{2} &1&0&0 &a_{2}\\
 0& a_{3} &1& a_{3} &0\\
 0& a_{4}&a_{4} &1 &0\\
 a_5 &0 & a_5 & a_5 &1 \\
\end{array}%
\right]\right) \label{eq:ex2_coded}. 
\end{align}
From (\ref{eq:ex2_uncoded}) and (\ref{eq:ex2_coded}), we observe that the two minrank optimization problems are equivalent, provided that, in (\ref{eq:ex2_uncoded}), the additional constraints are imposed on indeterminate elements per each rows such that $a_k^j=a_k^i$ for $j\neq i$. Intuitively, for the uncoded case, we are able to exploit different side-information separately in decoding, which provides more degrees of freedom to choose the indeterminate elements. For the coded case, however, the side-information is only used in the coded fashion in decoding, which imposes the constraints on the indeterminate elements.

\textbf {Example 3 (Coding over Side-Information in a Caching 
Problem):} Suppose the case of $N=K=2$ in which a server intends to deliver two files with file size $F=2$, i.e., ${\bf x}_1=[a_1, a_2]^T$ and ${\bf x}_2=[b_1,b_2]^T$ to the user 1 and 2, each with one bit cache memory $FM_k=1$. The coded caching method proposed in \cite{MN} is to store $a_1+b_1$ and $a_2+b_2$ to user 1 and user 2 in the caching phase. During the delivery phase, the transmitter sends $b_1$ and $a_2$ over two channel uses to satisfy the users' request. 

 In our framework, this caching method can be realized by choosing the side-information matrices such that
\begin{align}
{\bf s}_1=[1~ 0~ 1~ 0]~~ \rm{and}~~{\bf s}_2=[0~ 1~ 0~ 1].
\end{align}
Since user 1 and user 2 desire to decode file ${\bf x}_1$ and ${\bf x}_2$, the requested packet index matrices are 
\begin{align}
{\bf I}_1=\left[%
\begin{array}{cccc}
 1&0&0&0 \\
 0&1&0&0 \\
\end{array}%
\right] ~~{\rm and}~~
{\bf I}_2=\left[%
\begin{array}{cccc}
 0 &0 &1&0 \\
 0 &0 &0&1 \\
\end{array}%
\right].
\end{align}
By selecting the indeterminate elements as ${\bf A}_1^T=[1, 0]^T$ and ${\bf A}_2^T=[0, 1]^T$, we obtain the minimum index code length with two because \begin{align}
\beta^{\star}_2(\{{\bf s}_k,{\bf I}_k\}_{k=1}^2)&={\rm rk}\left(\!\!\left[%
\begin{array}{c}
 {\bf I}_1+{\bf A}_1^T{\bf s}_1 \\ {\bf I}_2+{\bf A}_2^T{\bf s}_2
\end{array}%
\right]\right) \nonumber \\ 
&={\rm rk}\left(\left[%
\begin{array}{cccc}
 0&0&1&0 \\
 0&1&0&0 \\
 0&0&1&0 \\
 0&1&0 &0 \\
\end{array}%
\right]\right) =2. \label{eq:example_caching_twouser}
\end{align}
As shown in the rows of the resultant matrix in (\ref{eq:example_caching_twouser}), there are two possible transmission schemes that satisfy the users' requests within two channel uses. The methods are to choose the index coding matrix ${\bf C}_{\rm IC}$ as either the first two rows in (\ref{eq:example_caching_twouser}) or the third and last rows in (\ref{eq:example_caching_twouser}).

\section{Algorithm via Matrix Completion}
In this section, we propose a index code design algorithm that leverages the minrank expression derived in Theorem 1. 
From Theorem 1, we observed that the optimal linear index coding matrix can be obtained by solving a matrix completion problem over a finite field. It is notable that the matrix competition problem to minimize the rank of the resultant matrix in (\ref{eq:Th1}) is different from the conventional matrix completion problems in \cite{Draper,Candes}. This discrepancy comes from the fact that, in our problem, an indeterminate element in $\{{\bf A}_k^{\star}\}$ affects the multiple entries in the resultant matrix in (\ref{eq:Th1}), which is not the case for the conventional matrix completion problems.

Using Theorem 1, we propose a greedy random search algorithm that finds a linear index coding matrix with the rank of $\beta$. The proposed algorithm initially  computes the rank of the matrix over $\mathbb{F}_2$, assuming that ${\bf A}_k^0$ is a matrix of zeros. Then, the algorithm runs over  until no rank change is detected over $U$ iterations in sequence. For the $p$-th iteration, we first update the elements of ${\bf A}_k^p$ randomly according to Bernoulli distribution with parameter $T$. Subsequently, we compute the rank of the matrix with the updated matrices $\{{\bf A}_k^p\}$ and store them if the rank decreases compared to the previous one. The algorithm is summarized in Table \ref{tab:greedy_alg}. In the proposed algorithm, the number of iterations $U$ plays a role in balancing between the performance and complexity. On the one hand, when the number of iterations $U$ is chosen sufficiently large, the proposed algorithm is able to yield the optimal minimum rank with high probability. On the other hand, when $U$ is not large enough, the probability that the algorithm reaches the minimum rank becomes low by reducing the computational complexity. Furthermore, the parameter $T$ controls the likelihood that side-information matrix $\{{\bf S}_k\}$ is used in decoding. This is because indeterminate elements in $\{{\bf A}_k\}$ is multiplied with the elements in side-information matrix ${\bf S}_k$. 

\begin{table}
\caption{The proposed greedy randomized index code design algorithm} 
\centerline{\scalebox{0.950}{
     \begin{tabular}{c|c}
	\hline
	 & Algorithm  \\
	\hline
	Input       & $\{{\bf R}_{k}\}$, $\{{\bf S}_{k}\}$ for $k\in\mathcal{K}$, $U$, and $T$\\
	Output       & ${\bf C}_{{\rm IC}}$\\
		\hline
	Initialization       & Set ${\bf A}^{\star}_k={\bf 0}$ for $k\in\mathcal{K}$\\
	& Set $p:=0$, $u:=0$, and $T:=t$;\\
	        & Compute $\beta=:{\rm rk}\left(\!\!\left[%
\begin{array}{c}
 {\bf R}_1+{{\bf A}_1^{\star}}^T{\bf S}_1 \\ \vdots \\ {\bf R}_K+{{\bf A}_K^{\star}}^T{\bf S}_K
\end{array}%
\right]\right)$  \\ 	\hline
	While &  $u< U$  \\
	&  \!\!\!\!\!\!\!\!\!\!\!\!\!\!\!\!\!\!\!\!\!\!\!\!\!\!  \!\!\!\!\!\!\!\!\!\!\!\!\!\!\!\!\!\!\!\!\!\!\!\!\!\! \!\!\!\!\!\!\!\!\!\!\!\!\!\!\!\!\!\!\!\!\!\!\!\!\!\! \!\!\!\!\!\!\!\!\!\!\!\!\!\!\!\!\!\!\!\!\!\!\!\!\!\!\!\!\!1. $p=:p+1$ \\
	& \!\!\!\!\!\!\!\!\!\!\!\!\!\!\!\!\!\!\!\!\!\!\!\!\!\! 2. Update  ${\bf A}^p_{k}(i,j)={\rm rand} > T$ for $\forall k,i,j$ \\
	&\!\!\!\!\!\!\!\!\!\!\!\!\!\! 3. Compute $r_p=:{\rm rk}\left(\!\!\left[%
\begin{array}{c}
 {\bf R}_1+{{\bf A}_1^{p}}^T{\bf S}_1 \\ \vdots \\ {\bf R}_K+{{\bf A}_K^p}^T{\bf S}_K
\end{array}%
\right]\right)$ \\
	 & \!\!\!\!\!\!\! 4. Update $\beta=:r_p$, ${\bf A}^{\star}_k:={\bf A}^p_k$, and $u=0$, if $r_p< \beta$ \\
	 &     \!\!\!\!\!\! \!\!\!\!\!\!\!\!\!\!\!\!\!\!\!\!\!\!\!\!\!\!\!\!\!\!\!\!\!\!\!\!\!\!\!\!\!\!\!\!\!\!\!\!\!\!\!\!Update $u=:u+1$, if $r_p\geq \beta$ \\
	end &\\
	\hline 
	 & ${\bf C}_{\rm{IC}}$=: Select $\beta$ independent rows in $\!\left[\!\!\!\!%
\begin{array}{c}
 {\bf R}_1+{{\bf A}_1^{\star}}^T{\bf S}_1 \\ \vdots \\ {\bf R}_K\!+\!{{\bf A}_K^{\!\star}}^T{\bf S}_K
\end{array}%
\!\!\!\!\right]$  \\	\hline
    \end{tabular}}}
    \label{tab:greedy_alg}
\end{table}

To verify the performance of the proposed algorithm. we consider the coded side-information case in Example 2.
\begin{align}
\beta_2^{\rm {coded}}=\!\min_{{a}_1,\ldots,a_5}\!\! {\rm rk}\left(\small\left[%
\begin{array}{ccccc}
 1 & a_{1} &0 &0 & a_{1} \\
 a_{2} &1&0&0 &a_{2}\\
 0& a_{3} &1& a_{3} &0\\
 0& a_{4}&a_{4} &1 &0\\
 a_5 &0 & a_5 & a_5 &1 \\
\end{array}%
\right]\right) \label{eq:ex3}. 
\end{align}
Note that the minimum rank of the matrix in (\ref{eq:ex3}) is two, which is obtained when every side-information is used in decoding, i.e., $a_k=1$ for $k\in\{1,\ldots,5\}$. Applying the proposed random greedy search method, we are able to compute the average of the rank $\mathbb{E}[\beta]$ as a function of the number of iterations $U$. As illustrated in Fig. \ref{fig:2}, when $T=0.1$ (the probability that each user does not use coded side-information in decoding), the proposed  algorithm achieves the minrank of 2 almost surely by searching $U=3$ points randomly among $2^5$ full search space. Furthermore, once the solutions $a_k=1$ are determined, we obtain the index coding matrix by arbitrary selecting two independent rows (2 and 3) in the matrix (\ref{eq:ex3}). As a result, it is possible for all users to decode $x_k$ if the transmitter sends $x_1+x_2+x_5$ and $x_2+x_3+x_4$ with two channel uses.

\begin{figure}
\centering \label{Fig2}
\includegraphics[width=3.3in]{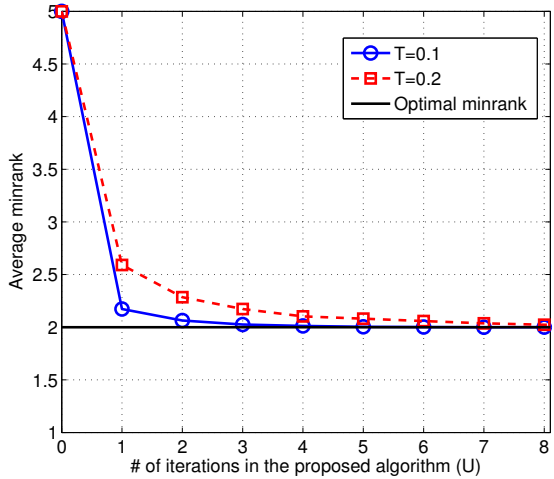}
\caption{Average minrank of the proposed algorithm accoding to the different parameters of $U$ and $T$.} \label{fig:2} 
\end{figure}

\section{Conclusion}
In this letter, we studied a class of index coding problems with coded side-information. The optimal binary linear index code length is characterized in terms of the minrank expression of a matrix using an algebraic approach. By leveraging the derived minrank expression, we proposed a simple algorithm that solves a matrix completion problem to design index codes. The analytical minrank expression derived in this letter can be applied to design caching algorithms in many content distribution systems.

\end{document}